\documentclass[11pt]{article}
\usepackage[margin=2.5cm]{geometry}
\usepackage{amsmath, amsfonts, amssymb, amsthm}
\usepackage{enumerate}
\usepackage{caption}
\usepackage{comment}
\usepackage[usenames,dvipsnames]{xcolor}
\usepackage[colorlinks=true, urlcolor=BrickRed, citecolor=MidnightBlue, linkcolor=Blue, filecolor=Blue]{hyperref}
\usepackage{algorithm, algorithmic}

\usepackage{url}

\numberwithin{equation}{section}
\newtheorem{theorem}{Theorem}
\newtheorem{lemma}{Lemma}
\newtheorem{corollary}{Corollary}
\newtheorem{definition}{Definition}[section]
\newtheorem{example}{Example}

\def\a{\alpha}

\def\t{\tau}

\def\e{\epsilon}
\def\d{\delta}

\def\l{\lambda}

\def\g{\gamma}
\def\hx{\hat{x}}
\def\bx{{\bf x}}
\def\by{{\bf y}}
\def\ttr{\texttt{RED}}
\def\ttb{\texttt{BLUE}}
\def\var{\text{Var}}

\def\ie{\emph{i.e.}}
\def\eg{\emph{e.g.}}
\def\etal{\emph{et al. }}

\newcommand{\expect}{\mathbb{E}}

\newcommand{\prob}{\mathbb{P}}
\newcommand{\real}{\mathbb{R}}

\def\beq{\begin {equation}}
\def\eeq{\end {equation}}
\def\beqar{\begin {eqnarray*}}
\def\eeqar{\end {eqnarray*}}
\def\beqarn{\begin {eqnarray}}
\def\eeqarn{\end {eqnarray}}

\newtheorem{thm}{Theorem}
\newtheorem{mydef}{Definition}[section]
\newtheorem{cor}{Corollary}[section]
\newtheorem{lem}{Lemma}[section]
\newtheorem{prop}{Proposition}[section]

\title{Biased Assimilation, Homophily and the Dynamics of Polarization}
\date{\today}
\author{
Pranav Dandekar\footnote{Department of Management Science \& Engineering, Stanford University, Stanford, CA. Email: \href{mailto:ppd@stanford.edu}{ppd@stanford.edu}.}
\and
Ashish Goel\footnote{Departments of Management Science \& Engineering and (by courtesy) Computer Science, Stanford University, Stanford, CA. Email: \href{mailto:ashishg@stanford.edu}{ashishg@stanford.edu}. }
\and
David Lee\footnote{Department of Electrical Engineering, Stanford University, Stanford, CA. Email: \href{mailto:davidtlee@stanford.edu}{davidtlee@stanford.edu}.}
}

\begin{document}
\maketitle
\begin{abstract}
Are we as a society getting more polarized, and if so, why? 
We try to answer this question through a model of opinion formation.
Empirical studies have shown that homophily results in polarization.
However, we show that DeGroot's well-known model of opinion formation based on repeated averaging can never be polarizing, even if individuals are arbitrarily homophilous. 
We generalize DeGroot's model to account for a phenomenon well-known in social psychology as \textit{biased assimilation}: when presented with mixed or inconclusive evidence on a complex issue, individuals draw undue support for their initial position thereby arriving at a more extreme opinion. 
We show that in a simple model of homophilous networks, our biased opinion formation process results in either polarization, persistent disagreement or consensus depending on how biased individuals are.
In other words, homophily alone, without biased assimilation, is not sufficient to polarize society.
Quite interestingly, biased assimilation also provides insight into the following related question: do internet based recommender algorithms that show us personalized content contribute to polarization?
We make a connection between biased assimilation and the polarizing effects of some random-walk based recommender algorithms that are similar in spirit to some commonly used recommender algorithms.
\end{abstract}

\newpage

\section{Introduction}
The issue of polarization in society has been extensively studied and vigorously debated in the academic literature as well as the popular press over the last few decades.
In particular, are we as a society getting more polarized, if so, why, and how can we fix it?
Different empirical studies arrive at different answers to this question depending on the context and the metric used to measure polarization.

Evidence of polarization in politics has been found in the increasingly partisan voting patterns of the members of Congress \cite{poole-1984:polarization-politics, poole-1991:congressional-voting} and in the extreme policies adopted by candidates for political office \cite{hill:divided}.
McCarty \etal \cite{poole:polarized} claim via rigorous analysis that America is polarized in terms of political attitudes and beliefs.
Phenomena such as segregation in urban residential neighborhoods (\cite{schelling-1971:segregation, bruch-2006:neighborhood-choices, bobby-kleinberg:schelling}), the rising popularity of overtly partisan television news networks \cite{nyt:cnn2, nyt:cnn1}, and the readership and linking patterns of blogs  along partisan lines \cite{Adamic-2005:blog,hargittai-2007:cross-bloggers,gilbert-2009:echo-chambers,Lawrence-2010:self-segregation} can all be viewed as further evidence of polarization.
On the other hand, it has also been argued on the basis of detailed surveys of public opinion that society as a whole is not polarized, even though the media and the politicians make it seem so \cite{wolfe:one-nation, fiorina:culture-wars}.
We adopt the view that polarization is not a property of a state of society; instead it is a property of the dynamics of interaction between individuals.

It has been argued that homophily, \ie, greater interaction with like-minded individuals, results in polarization \cite{baron-1996:corroboration, sunstein:republic.com, gilbert-2009:echo-chambers}. 
Evidence in support of this argument has been used to claim that the rise of cable news, talk radio and the Internet has contributed to polarization: the increased diversity of information sources coupled with the increased ability to narrowly tailor them to one's specific tastes (either manually or algorithmically through, for example, recommender systems)  has an echo-chamber effect which ultimately results in increased polarization.

%The polarization question is closely tied to the work on understanding the dynamics of interaction and influence among people.
%An early mathematical model proposed by \cite{cartwright-harary:structural-balance} explained the rise of factions in a network using the theory of structural balance.
%However, it is based on triadic interactions which are more complex and less plausible compared to dyadic ones.
%Mathematical models based on other theories of interpersonal dynamics from social psychology (\eg, \cite{nowak-1990:private-public, baldassarri-2007:political-polarization}) have also been used to explain polarization.

A rich body of work attempts to explain polarization through variants of a well-known opinion formation model due to DeGroot \cite{degroot:consensus}.
In DeGroot's model, individuals are connected to each other in a social network.
The edges of the network have associated weights representing the extent to which neighbors influence each other's opinions.
Individuals update their opinion as a weighted average of their current opinion and that of their neighbors.
Variants of this model (\eg, \cite{friedkin-johnsen, krause:bounded-conf, ace-oz:fluctuations, kleinberg-opinion}) account for the empirical observation that in many cases there is persistent disagreement between individuals and consensus is never reached.
However, we show that repeated averaging of opinions, which underlies these models, always results in opinions that are less divergent compared to the initial opinions, even if individuals are arbitrarily homophilous.
As a result, this entire body of work appears to fall short of explaining polarization which is generally perceived to mean an \textit{increased} divergence of opinions, not just persistent disagreement.
In this paper, we seek a more satisfactory model of opinion formation that (a) is informed by a theory of how individuals actually form opinions, and (b) produces an increased divergence of opinions under intuitive conditions.

We base our model on a well-known phenomenon in social psychology called \textit{biased assimilation}, according to which individuals process new information in a biased manner whereby they readily accept confirming evidence while critically examining disconfirming evidence.
Suppose that individuals with opposing views on an issue are shown mixed or inconclusive evidence.
Intuitively, exposure to such evidence would engender greater agreement, or at least a moderation of views. 
However, in a seminal paper, Lord \etal \cite{ross-lepper} showed through experiments that biased assimilation causes individuals to arrive at \textit{more extreme} opinions after being exposed to \textit{identical}, inconclusive evidence.
This finding has been reproduced in many different settings over the years (\eg, \cite{miller:attitude-polarization, munro:1996-prez-debate, taber:motivated-skepticism}).
We use biased assimilation as the basis of our model of opinion formation and show that in our model homophily alone, without biased assimilation, is not sufficient to polarize society.

\subsection{Summary of Contributions}
We propose a generalization of DeGroot's model that accounts for biased assimilation.
Like DeGroot's model, our opinion formation process unfolds over an exogenously defined social network represented by a weighted undirected graph $G  = (V,E)$.
Each individual $i\in V$ has an opinion $x_i(t) \in [0,1]$, which represents his degree of support at time step $t$ for the position represented by 1.
In order to weight confirming evidence more heavily relative to disconfirming evidence, opinions are updated as follows:
individual $i$ weights each neighbor $j$'s opinion $x_j(t)$ by a factor $(x_i(t))^{b_i}$ and weights the opposing view $(1-x_j(t))$ by a factor $(1-x_i(t))^{b_i}$, where $b_i\ge 0$ is a \textit{bias parameter}. 
Informally, $b_i$ represents the bias with which $i$ assimilates his neighbors opinions.
When $b_i =0$, our model reduces to DeGroot's, and corresponds to unbiased assimilation. 
Our biased opinion formation process mathematically reproduces the effect empirically observed by Lord \etal (Theorem~\ref{thm:single-agent}).

We measure divergence of opinions in terms of the \textit{network disagreement index} (NDI), which we define to be $\sum_{(i,j)\in E} w_{ij}(x_i(t) - x_j(t))^2$.
It is similar to the notion of social cost used by Bindel \etal \cite{kleinberg-opinion}.
We say that an opinion formation process is \textit{polarizing} if the NDI at the end of the process is greater than that initially.
We show that:
\begin{itemize}
\item
(Theorem~\ref{thm:degroot-not-polarizing}) DeGroot-like repeated averaging processes can never be polarizing, even if individuals are arbitrarily homophilous (\ie, the underlying network is presented adversarially as opposed to based on a mathematical model).
\item
(Theorem~\ref{thm:homophily-polarization}) The biased opinion formation process over a simple model of networks with homophily results in polarization if  individuals' bias parameter $b\ge 1$. If $b < 1$, the process results in either persistent disagreement or consensus depending on the degree of homophily.
\end{itemize}
In summary, we show that homophily alone, without biased assimilation, is not sufficient to polarize society.
This conclusion disagrees with the literature (\eg, \cite{baron-1996:corroboration, sunstein:republic.com}) that proposes homophily as the predominant cause of polarization. 
As the reader might expect, there are many ways of mathematically measuring the divergence of opinions among individuals. 
Many of our results hold for more general measures of divergence, which we discuss in Section~\ref{sec:polarization-metrics}.

The notion of biased assimilation also provides insight into the following related question: do internet based recommender algorithms that show us personalized content contribute to polarization?
We analyze the polarizing effects of  three recommender algorithms---SimpleSALSA, SimplePPR, and SimpleICF---that are similar in spirit to three well-known algorithms from the literature: SALSA \cite{salsa}, Personalized PageRank \cite{pagerank}, and Item-based Collaborative Filtering \cite{item-cf}.
For a simple, natural model of the underlying user-item graph, and under reasonable assumptions, we show that SimplePPR, which recommends the item that is most relevant to a user based on a PageRank-like score, is always polarizing (Theorem~\ref{thm:ppr-n->infty}).
On the other hand, SimpleSALSA and SimpleICF, which first choose a random item liked by the user and recommend an item similar to that item, 
are polarizing only if individuals are biased (Theorem~\ref{thm:salsa-icf-n->infty}).
Designing algorithms and online social systems that reduce polarization, for example, by counteracting biased assimilation is a promising research direction.

\section{Model}
Our opinion formation process unfolds over a social network represented by a \textit{connected weighted undirected} graph $G = (V,E,w)$.
The nodes in $V$ represent individuals and the edges represent friendships or relationships between them.
Let $|V| = n$. 
An edge $(i,j)\in E$ is associated with a weight $w_{ij} > 0$ representing the degree of influence $i$ and $j$ have on each other.
Each individual $i\in V$ also has an associated weight $w_{ii} \ge 0$ representing the degree to which the individual weights his own opinions.
We will denote by $N(i)$ the set of neighbors of $i$, that is, $N(i) :=  \{ j \in V : (i,j)\in E\}$.

An individual $i$ has an opinion $x_i(t)\in [0,1]$ at time step $t = 0, 1, 2, \dotsc$.
The extreme opinions 0 and 1 represent two opposing points of view on an issue. 
So $x_i(t)$ can be interpreted as individual $i$'s degree of support at time $t$ for the position represented by 1, and $1-x_i(t)$ as the degree of support for the position represented by 0.
Let $\bx(t)\in [0,1]^n$ denote the vector of opinions at time $t$.
An opinion formation process is simply a description of how individuals update their opinions, \ie, for each individual $i\in V$, it defines $x_i(t+1)$ as a function of the vector of opinions, $\bx(t)$, at time $t$.

\subsection{Measuring Polarization}
We view polarization as a property of an opinion formation process instead of a property of a state of the network.
We characterize polarization as a \textit{verb} as opposed to a \textit{noun}, \ie, we say that an opinion formation process is \textit{polarizing} if it results in an increased divergence of opinions. One could mathematically capture divergence of opinions in many different ways. We measure divergence in terms of the \textit{network disagreement index} defined below.
\begin{definition}[Network Disagreement Index (NDI)]
Given a graph $G = (V,E,w)$ and a vector of opinions $\bx\in [0,1]^n$ of individuals in $V$, the \emph{network disagreement index} $\eta(G,\bx)$ is defined as
\beq\label{def:eta}
\eta(G, \bx) := \sum_{(i,j)\in E} w_{ij}(x_i - x_j)^2
\eeq
\end{definition}
Consider an opinion formation process over a network $G = (V,E,w)$ that transforms a set of initial opinions $\bx\in [0,1]^n$ into a set of opinions $\bx'\in[0,1]^n$. Then, we say the process is polarizing if $\eta(G,\bx') > \eta(G,\bx)$, and vice versa.

The NDI is similar to the notion of social cost used by Bindel \etal \cite{kleinberg-opinion}.
Each term $w_{ij}(x_i-x_j)^2$ can be viewed as the cost of disagreement imposed upon $i$ and $j$.
This view that the social cost depends on the magnitude of the difference of opinions along edges is consistent with theories in social psychology according to which attitude conflicts in relationships are a source of psychological stress or instability \cite{heider:balance, festinger:cognitive}.
The NDI  captures the phenomenon of \textit{issue radicalization}, \ie, pre-existing groups of individuals becoming progressively more extreme.
Admittedly, it does not entirely capture an aspect of polarization called \textit{issue alignment} \cite{baldassarri-gelman} whereby individuals with diverse opinions organize into ideologically coherent, but opposing factions.
However, there is significant empirical evidence \cite{poole:polarized, baldassarri-gelman, nyt:sunstein} that issue radicalization is more prevalent compared to issue alignment, and hence NDI captures the most salient aspects of polarization.
Many of our results hold for more general measures of divergence which we discuss in Section~\ref{sec:polarization-metrics}.

%Our main result, which we present in Section~\ref{sec:degroot} and Section~\ref{sec:homophily}, is that the biased opinion formation process is polarizing if the individuals are (a) sufficiently biased (\ie, $b_i \ge 1$), and (b) homophilous, whereas DeGroot's repeated averaging process is always depolarizing, even in the presence of homophily.

\subsection{DeGroot's Repeated Averaging Model}
In his seminal work on opinion formation, DeGroot \cite{degroot:consensus} proposed a model where at each time step, individuals simultaneously update their opinion to the weighted average of their neighbors' and their own opinion at the previous time step.
 \begin{definition}[DeGroot's Repeated Averaging Process]
The opinion of individual $i$ at time $t+1$, $x_i(t+1)$, is given by
\beq\label{def:degroot}
x_i(t+1) = \frac{w_{ii}x_i(t) + s_i(t)}{w_{ii} + d_i}
\eeq
\end{definition}
where $s_i(t) := \sum_{j\in N(i)} w_{ij}x_j(t)$ is the weighted sum of the opinions of $i$'s neighbors,  and $d_i := \sum_{j\in N(i)} w_{ij}$ is $i$'s weighted degree.

Recall that $x_j(t)$ and $1-x_j(t)$ represent the degree of support for extremes $1$ and $0$, respectively. 
Then, opinion update under DeGroot's process is equivalent to taking a weighted average of the  total support for 0 and that for 1.
The weight that individual $i$ places on 1 (and on 0) is computed by summing the degrees of support of $i$'s neighbors weighted by the influence of each neighbor on $i$.

\subsection{Biased Opinion Formation Model}\label{subsec:biased-model}
We generalize DeGroot's model to account for \textit{biased assimilation}.
Biased assimilation is a well-known phenomenon in social psychology described by Lord \etal \cite{ross-lepper} in their seminal paper as follows: 
\begin{quote}
\small
People who  hold  strong opinions on  complex social issues  are likely to examine relevant empirical evidence in  a  biased  manner.  They  are  apt to  accept ``confirming"  evidence  at  face value  while  subjecting  ``disconfirming"  evidence  to critical evaluation, and  as a result to  draw undue support  for  their initial positions  from  mixed  or  random empirical  findings.
\end{quote}
Lord \etal \cite{ross-lepper} showed through experiments that biased assimilation of mixed or inconclusive evidence does indeed result in more extreme opinions.

In order to account for biased assimilation, we propose a \textit{biased opinion formation process}.
Recall that $x_i(t)$ can be viewed as the degree of support for the position represented by 1.
Individuals weight confirming evidence more heavily relative to disconfirming evidence by updating their opinions as follows:
individual $i$ weights each neighbor $j$'s support for 1 (\ie, $x_j(t)$) by an additional factor $(x_i(t))^{b_i}$,  where $b_i\ge 0$ is a \textit{bias parameter}. 
Therefore, $x_i(t+1) \propto (x_i(t))^{b_i}w_{ij}x_j(t)$, Similarly, $i$ weights $j$'s support for 0 (\ie, $1-x_j(t)$) by $(1-x_i(t))^{b_i}$, and so $(1-x_i(t+1)) \propto (1-x_i(t))^{b_i}w_{ij}(1-x_j(t))$.
Informally, $b_i$ represents the bias with which $i$ assimilates his neighbors opinions.

{\bf Illustrative example}. Consider a graph with two nodes, $i$ and $j$, connected by an edge with a weight $w_{ij}$. Then, according to the biased opinion formation process, $i$'s opinion at time $t+1$, $x_i(t+1)$, is given by
\[
x_i(t+1) = \frac{w_{ii}x_i(t) + (x_i(t))^{b_i}w_{ij}x_j(t)}{w_{ii} + (x_i(t))^{b_i}w_{ij}x_j(t) + (1-x_i(t))^{b_i}w_{ij}(1 - x_j(t))}
\]
More generally, the opinion update of individual $i$ in the biased opinion formation process is defined as below.
\begin{definition}[Biased Opinion Formation Process]
Under the biased opinion formation process, the opinion of individual $i$ at time $t+1$, $x_i(t+1)$, is given by
\beq\label{def:biased-update}
x_i(t+1) = \frac{w_{ii}x_i(t) + (x_i(t))^{b_i}s_i(t)}{w_{ii} + (x_i(t))^{b_i}s_i(t) + (1-x_i(t))^{b_i}(d_i - s_i(t))}
\eeq
\end{definition}
where, as before, $s_i(t) := \sum_{j\in N(i)} w_{ij}x_j(t)$ is the weighted sum of the opinions of $i$'s neighbors,  and $d_i := \sum_{j\in N(i)} w_{ij}$ is $i$'s weighted degree.
Observe that when $b_i = 0$, \eqref{def:biased-update} is identical to \eqref{def:degroot}, \ie, DeGroot's averaging process is a special case of our process and corresponds to unbiased assimilation.
More generally, biased assimilation can be modeled by making $i$'s opinion update proportional to $\beta_i(x_i(t))s_i(t)$, where the bias function $\beta_i: [0,1]\rightarrow [0,1]$ is non-decreasing.

{\bf Connection with Urn Models.} Urn models are an elegant abstraction that have been used to analyze the properties of a wide variety of probabilistic processes. DeGroot's model of weighted averaging has the following analogous urn dynamic: $x_i(t)$ denotes the fraction of \ttr\ balls in individual $i$'s urn at time $t$, and $1-x_i(t)$ denotes the corresponding fraction of \ttb\ balls. At each time step, $i$ chooses a neighbor $j$ with probability proportional to $w_{ij}$ and chooses a ball uniformly at random from $j$'s urn. Individual $i$ adds that ball to his urn and discards a ball chosen uniformly  at random from his urn.
When the bias parameter $b_i = 1$, the biased opinion formation process can be interpreted as the following variant of the above urn dynamic: as before, $i$ chooses a neighbor $j$ with probability proportional to $w_{ij}$ and chooses a ball uniformly at random from $j$'s urn. In addition, $i$ also chooses a ball uniformly at random from his own urn. If the colors of the two balls match, $i$ puts them both into his urn and discards a ball chosen uniformly  at random from his urn. If the colors do not match, the two balls are returned to their respective urns.

\subsection{Biased Assimilation by a Single Agent in a Fixed Environment}
Here we demonstrate that our model of biased assimilation mathematically reproduces the empirical findings of Lord \etal \cite{ross-lepper}.
We analyze the change in opinion of a single individual as a function of his bias parameter when he is exposed to opinions from a \textit{fixed}  environment.
The fixed environment represents sources of information that influence the individual's opinion, but can be assumed to remain unaffected  by the individual's opinion, such as the news media, the Internet, the organizations that the individual is a part of, etc.

For this section, we will denote by $x(t)\in [0,1]$ the individual's opinion at time $t$, and by $b\ge0$ the individual's bias parameter. Let the individual's weight on his own opinion, $w_{ii} = w$. Let $s\in (0,1)$ denote the (time-invariant) weighted average of the opinions of all sources in the individual's environment. Then, from \eqref{def:biased-update}, the individual's opinion at time $t+1$ is given by
\beq\label{def:single-ind-update}
x(t+1) = \frac{wx(t) + (x(t))^bs}{w + (x(t))^bs + (1-x(t))^b(1-s)}
\eeq

Given $s\in (0,1)$, and $b\ne 1$, we define
\beq\label{def:polarization-threshold}
\hx(s,b) := \frac{s^{1/(1-b)}}{s^{1/(1-b)} + (1-s)^{1/(1-b)}}
\eeq
as the \textit{polarization threshold} for the individual. 
We show that when the individual is sufficiently biased (\ie, $b > 1$), the polarization threshold $\hx$ is an unstable equilibrium, \ie, in equilibrium the individual's opinion goes to 1 or 0 depending on whether the initial opinion was greater than or less than $\hx$.
On the other hand, when $b < 1$, $\hx$ is a stable equilibrium.
\begin{thm}\label{thm:single-agent}
Fix $t \ge 0$. Let $x(t) \in (0,1)$.
\begin{enumerate}
\item
If $b> 1$,
\begin{enumerate}
\item
 if $x(t) > \hx$, then $x(t+1) > x(t)$, and $x(t)\rightarrow 1$ as $t\rightarrow \infty$. 
 \item
 if $x(t) <\hx$, then $x(t+1) < x(t)$, and $x(t)\rightarrow 0$ as $t\rightarrow \infty$. 
 \item
 if $x(t) = \hx$, then for all $t'> t, x(t') = \hx$.
 \end{enumerate}
 \item
 If $b < 1$, 
 \begin{enumerate}
\item
 if $x(t) > \hx$, then $x(t+1) < x(t)$. 
 \item
 if $x(t) < \hx$, then $x(t+1) > x(t)$.
 \item
 $x(t)\rightarrow \hx$ as $t\rightarrow \infty$.
  \end{enumerate}
 
\end{enumerate}
\end{thm}
The theorem is proved in Appendix~\ref{app:proofs-sec2}.
The opinion $x(t)$ can be interpreted as the individual's degree of support  for the extreme represented by 1. 
So, the above theorem shows that when the individual is sufficiently biased (\ie, $b > 1$), exposure to the environment pushes him away from the threshold $\hx$ (unless $x(0) = \hx$), and toward one of the extremes, and the individual holds an extreme opinion ($x(t) = 0$ or $x(t) = 1$) in equilibrium. Thus $\hx$ is an unstable equilibrium.
This mathematically captures the biased assimilation behavior observed empirically.
On the other hand, if the individual has low bias (\ie, $b < 1$), then he gravitates towards the polarization threshold $\hx$ over time. Thus, $\hx$ is a stable equilibrium in this case.
The behavior of the individual when $b = 1$ is a limiting case of the two cases proven in the theorem; as $b\rightarrow 1,\ \hx \rightarrow s$.
When the individual is connected to other individuals in a social network, we will show that the biased opinion formation process produces polarization even when $b= 1$.

\section{DeGroot's Repeated Averaging Process is not Polarizing}\label{sec:degroot}
It is easy to see that if DeGroot's process was asynchronous, \ie, individuals update their opinion one at a time, each opinion update can only lower the network disagreement index (NDI).
However, here we will show that each opinion update can only lower the NDI even when individuals update opinions simultaneously.
As a result, the repeated averaging process is depolarizing.
Our result holds for arbitrary weights $w_{ij}$, and an arbitrary vector of opinions $\bx\in [0,1]^n$, \ie, when the underlying network is arbitrarily homophilous.

\begin{thm}\label{thm:degroot-not-polarizing}
Consider an arbitrary weighted undirected graph $G = (V,E,w)$. Assume that $G$ is connected. Let $\bx(t)\in [0,1]^n$ be an arbitrary vector of opinions of nodes in $G$ at time $t\ge0$. Assume that for all $i\in V,\  b_i = 0$. Then, $\eta(G, \bx(t+1)) \le \eta(G, \bx(t))$, \ie, the network disagreement index at time $t+1$ is no more than that at time $t$.
\end{thm}

The theorem is proved in Appendix~\ref{app:proofs-sec3}.
Observe that in the limit as $w_{ii}\rightarrow\infty$, individual $i$ can be viewed as being a \textit{zealot} \cite{ace-oz:zealots},\ie, an individual with an unchanging opinion.
So our result also holds for repeated averaging in the presence of zealots.

A possible criticism of this result is that it holds for this particular definition of the NDI which may not always capture the intuitive notion of polarization. 
For example, consider a network partitioned into two densely connected opposing factions with sparse cross linkages. One might consider such a network to be polarized, even though the network disagreement index for it is small.
An alternate measure that does capture the divergence of opinions in the above example is the \textit{global disagreement index} (GDI) defined below.
\begin{definition}[Global Disagreement Index (GDI)]
Given a vector of opinions $\bx\in [0,1]^n$ of individuals in $V$, the \emph{global disagreement index} $\g(\bx)$ is defined as
\beq\label{def:gamma}
\g(\bx) := \sum_{i < j} (x_i - x_j)^2
\eeq
\end{definition}
Observe that it is possible to assign edge weights $w_{ij}$ such that DeGroot's repeated averaging process increases the GDI since the latter is independent of the weights.
However, we show that a variant of repeated averaging, based on the well-known flocking model for decentralized consensus \cite{jnt-phdthesis}, can only decrease the GDI. 
We consider a repeated  averaging process where at each time step $t\ge0$, an arbitrary set $S(t)\subseteq V$ of individuals simultaneously updates their opinions to be closer to the average opinion of the set.
\begin{definition}[Flocking Process]
Let $\e \in [0,1]$. For $t\ge 0$, let $S(t)\subseteq V$ be an arbitrary set of individuals. Let $s(t) := \frac1{|S(t)|}\sum_{i\in S(t)} x_i(t)$ be the average opinions of individuals in $S(t)$. Under the flocking process, the opinion of individual $i\in V$ at time $t+1$, $x_i(t+1)$, is given by
\beq\label{def:flocking-update}
x_i(t+1) = \left\{
\begin{array}{rl}
(1-\e)x_i(t) + \e s(t),&\ \text{ if }i\in S(t)\\
x_i(t),&\ \text{ otherwise}
\end{array}
\right.
\eeq
\end{definition}
Next we show that each opinion update in the flocking process can only lower the GDI.
\begin{thm}\label{thm:flocking-not-polarizing}
Let $\bx(t)\in [0,1]^n$ be an arbitrary vector of opinions of nodes in $V$ at time $t\ge0$. Let $\bx(t+1)\in [0,1]^n$ be the vector of opinions at time $t+1$ after one step of the flocking process. Then, $\g(\bx(t+1)) \le \g(\bx(t))$, \ie, the GDI at time $t+1$ is no more than that at time $t$.
\end{thm}
The theorem is proved in Appendix~\ref{app:proofs-sec3}.

\section{Polarization due to Biased Assimilation}\label{sec:homophily}
In this section we state and prove our main result: in a simple model of networks with homophily, the biased opinion formation process may result in either polarization, persistent disagreement, or consensus depending on how biased the individuals are.
We model homophilous networks using a deterministic variant of \emph{multi-type random networks} \cite{golub-jackson:homophily}. 
Multi-type random networks are a generalization of Erd\"{o}s-R\'{e}nyi random graphs. Nodes in $V$ are partitioned into \emph{types}, say, $\t_1, \t_2,\dotsc, \t_k$. 
The network is parameterized by a vector $(n_1, \dotsc, n_k)$ where $n_i$ is the number of nodes of type $\t_i$, and a symmetric matrix $P \in [0,1]^{k\times k}$, where $P_{ij}$ is the probability that there exists an undirected edge between a node of type $\t_i$ and another of type $\t_j$. 
The class of multi-type random networks where $P_{ii} > P_{ij}$ for all $i,j$, is called is the \text{islands model}, and is used to model homophily (since an individual is more likely to be connected with individuals of the same type).
We will analyze the biased opinion formation process over a deterministic variant of the islands model, which we call a \emph{two-island network}.
\begin{mydef}
Given integers $n_1, n_2 \ge 0$, and real numbers $p_s, p_d \in (0,1)$, a $(n_1, n_2, p_s, p_d)$-two island network is a weighted undirected graph $G = (V_1, V_2, E, w)$ where
\begin{itemize}
\item
$|V_1| = n_1, |V_2| = n_2$ and $V_1 \cap V_2 = \emptyset$.
\item
Each node $i\in V_1$ has $n_1p_s$ neighbors in $V_1$ and $n_2p_d$ neighbors in $V_2$.
\item
Each node $i\in V_2$ has $n_2p_s$ neighbors in $V_2$ and $n_1p_d$ neighbors in $V_1$\footnote{For clarity of exposition, we assume that the quantities $n_1p_s, n_2p_s, n_1p_d$ and $n_2p_d$ are all integers.}.
\item
$p_s > p_d$.
\end{itemize}
\end{mydef}
For a two-island network, we define the \emph{degree of homophily} as follows.
\begin{mydef}
Let $G = (V_1, V_2,E, w)$ be a $(n_1, n_2, p_s, p_d)$-two island network. Then the degree of homophily in $G$, $h_G$, is defined to be the ratio $p_s/p_d$.
\end{mydef}

Informally, a high value of $h_G$ implies that nodes in $V$ are much more likely to form edges to other nodes of their own type, thereby exhibiting a high degree of homophily.
\begin{thm}\label{thm:homophily-polarization}
Let $G = (V_1, V_2, E, w)$ be a $(n,n,p_s,p_d)$-two island network. For all $i\in V = V_1\cup V_2$, let $w_{ii} = 0$. For all $(i,j)\in E$, let $w_{ij} = 1$.  Assume for all $i\in V_1$, $x_i(0) = x_0$ where $\frac12 < x_0 < 1$. Assume for all $i\in V_2,\ x_i(0) = 1-x_0$. Assume for all $i\in V$, the bias parameter $b_i = b > 0$. Then, 
\begin{enumerate}
\item
(Polarization) If $b\ge 1,\ \forall i\in V_1,\ \lim_{t\rightarrow \infty} x_i(t) =  1$, and $\forall i\in V_2,\ \lim_{t\rightarrow \infty} x_i(t) =  0$.
\item
(Persistent Disagreement) if $1 > b \ge \frac2{h_G+1}$, then there exists a unique $\hx\in (\frac12,1)$ such that $\forall i\in V_1,\ \lim_{t\rightarrow \infty} x_i(t) =  \hx$, and $\forall i\in V_2,\ \lim_{t\rightarrow \infty} x_i(t) =  1-\hx$.
\item
(Consensus) if $b <\frac2{h_G+1}$, then for all $i\in V,\ \lim_{t\rightarrow \infty} x_i(t) =  \frac12$.
\end{enumerate}
\end{thm}
The theorem is proved in Appendix~\ref{app:proofs-sec4}.
Let us analyze the implications of this theorem. Let $\eta(G,\bx(t))\rightarrow \eta_{\infty}$ as $t\rightarrow\infty$, \ie, let $\eta_{\infty}$ be the NDI at equilibrium. Then, the above result implies that when $b\ge 1$, $\eta_{\infty} > \eta(G,\bx(0))$, \ie, the biased opinion formation process is polarizing. On the other hand, when individuals are moderately biased (\ie, $1 > b\ge 2/(h_G+1)$), $\eta_{\infty} > \eta(G,\bx(0))$ if and only if $x_0 < \hx$; so the opinion formation process may not be polarizing, but it doesn't produce consensus either. Finally, when individuals have low bias (\ie, $b < 2/(h_G+1)$, $\eta_{\infty} = 0 < \eta(G,\bx(0))$, \ie, the opinion formation process is depolarizing, since the network reaches consensus in equilibrium.

This illustrates the importance of the bias parameter in causing polarization.
Also, observe that $b=1$ corresponds to the urn dynamic described in Section~\ref{subsec:biased-model}, and hence the above result shows that that urn dynamic causes polarization for arbitrarily small degree of homophily.

\section{Recommender Systems and Polarization}
Recommender systems are widely used on the Internet to present personalized information (\eg, search results, new articles, products) to individuals. 
This personalization is typically done by algorithms that use an individual's the past behavior (\eg, history of browsing and purchases) and of other individuals that are similar in some way to that individual, to discover items of possible interest to the user.
It has been argued \cite{sunstein:republic.com} that this personalization of information has an echo-chamber effect where individuals are only exposed to information they agree with, and this ultimately leads to increased polarization.
In this section we investigate this question: do recommender systems have a polarizing effect?
We analyze three simple random-walk based recommender algorithms---SimpleSALSA (Algorithm~\ref{rec:salsa}), SimplePPR (Algorithm~\ref{rec:PPR}) and SimpleICF(Algorithm~\ref{rec:ICF})--- that are similar in spirit to three well-known recommender algorithms from the literature: SALSA \cite{salsa}, Personalized PageRank \cite{pagerank}, and item-based collaborative filtering \cite{item-cf}, respectively.

\begin{algorithm}[t]
\begin{algorithmic}
\STATE {\bf Input:} $G = (V_1, V_2, E)$, node $i\in V_1$.
\STATE Perform a three-step random walk on $G$ starting at $i$.
\STATE Let the random walk end at node $j\in V_2$.
\STATE {\bf Output:} $j$.
\end{algorithmic}
\caption{SimpleSALSA}
\label{rec:salsa}
\end{algorithm}

We consider the following simple model: 
Let $G = (V_1, V_2,E)$ be an unweighted undirected bipartite graph.  
Nodes in $V_1$ represent individuals. 
Nodes in $V_2$ represent items.
The items could be books, webpages, news articles, products, etc. 
For concreteness, we will refer to nodes in $V_2$ as books.
For a node $i\in V_1$ and a node $j\in V_2$, an edge $(i,j)\in E$ represents ownership, \ie, individual $i$ owns book $j$. 
For our purpose, we define a recommender algorithm as below.
\begin{definition}
A recommender algorithm takes as input a bipartite graph $G = (V_1, V_2, E)$ and a node $i\in V_1$, and outputs a node $j\in V_2$.
\end{definition}
Thus, given a graph representing which users own which books, and a specific user $i$, a recommender algorithm outputs a single book $j$ to be recommended to $i$.
We assume that $i$ can only buy a book if it is recommended to him.
However, he may choose to reject a recommendation, \ie, to not buy a recommended book.
Therefore, $i$ buying a book $j$ requires two steps: the recommender algorithm must recommend $j$ to $i$, and then $i$ must accept the recommendation.

Since, we are interested in analyzing the polarizing effects of recommender systems, we will assume that each book in $V_2$ is labeled either `\ttr' or `\ttb'.
These labels are purely for the purpose of analysis; the algorithms we study are agnostic to  these labels.
For each individual $i\in V_1$, let $x_i\in [0,1]$ be the fraction of \ttr\ books owned by $i$, and $1-x_i$ be that of \ttb\ books.
Individuals may be biased, or unbiased, as we define below.

 \begin{definition}
 Consider a book recommended to an individual $i\in V_1$. We say that $i$ is \textit{unbiased} if $i$ accepts the recommendation with the same probability independent of whether the book is \ttr\ or \ttb. We say that $i$ is \textit{biased} if 
 \begin{enumerate}
 \item
 $i$ accepts the recommendation of a \ttr\ book with probability $x_i$, and rejects it with probability $1-x_i$, and 
 \item
 $i$ accepts the recommendation of a \ttb\ book with probability $1-x_i$, and rejects it with probability $x_i$.
 \end{enumerate}
 \end{definition}
Observe that the above definition of an individual $i$ being biased corresponds to the urn dynamic described in Section~\ref{subsec:biased-model} with $b_i = 1$.
For an individual $i$, the fraction of \ttr\ books $i$ owns, $x_i$, can be viewed as $i$'s opinion in the interval $[0,1]$, and so a recommender algorithm can be viewed as an opinion formation process. 
The opinion $x_i$ remains unchanged if $i$ rejects a recommendation. 
However, if $i$ accepts a recommendation, $x_i$ increases or decreases depending on whether the recommended book was \ttr\ or \ttb. 
Thus, we are interested in the probability that a recommendation was for a \ttr\ (or \ttb) book \textit{given} that $i$ accepted the recommendation.
The above probability determines whether a recommender algorithm is polarizing or not.

\begin{definition}
Consider a recommender algorithm and an individual $i\in V_1$ that accepts the algorithm's recommendation. The algorithm is \textit{polarizing} with respect to $i$ if 
\begin{enumerate}
\item
when $x_i > \frac12$,  the probability that the recommended book was \ttr\ is greater than $x_i$, and
\item
when $x_i < \frac12$,  the probability that the recommended book was \ttr\ is less than $x_i$.
\end{enumerate}
\end{definition}

In order to analyze the recommender algorithms, we assume a generative model for $G$, which we describe next.

 \begin{algorithm}[t]
\begin{algorithmic}
\STATE {\bf Input:} $G = (V_1, V_2, E)$, node $i\in V_1$.
\STATE {\bf Parameter:} A large positive integer $T$.
\STATE Perform $T$ three-step random walks on $G$ starting at node $i$.
\STATE For node $j\in V_2$, let \texttt{count(j)} be the number of random walks that end at node $j$.
\STATE {\bf Output:} $j^* := \arg\max_j \texttt{count(j)}$.
\end{algorithmic}
\caption{SimplePPR}
\label{rec:PPR}
\end{algorithm}

\subsection{Generative Model for $G$}
Let the number of individuals, $|V_1| = m > 0$.
Let the number of books, $|V_2| = 2n$, with $n > 0$ books of each color.
We assume that $m = f(n)$; and $\lim_{n\rightarrow\infty} f(n) = \infty$. 
For each individual $i\in V_1$, we draw $x_i$ independently from a distribution over $[0,1]$ with a probability density function (pdf) $g(\cdot)$.
We assume that $g$ is symmetric about $\frac12$, \ie, for all $y\in [0,1],\ g(y) = g(1-y)$. 
This implies that for all $i\in V_1,\ \expect[x_i] = \frac12$.
We assume that the variance of the distribution is strictly positive, \ie, $\var(x_i) > 0$.
For an individual $i$ and a \ttr\ book $j$, there exists an edge $(i,j)\in E$ independently with probability $\frac{x_ik}{n}$, where $0 < k < n$.
For an individual $i$ and a \ttb\ book $j$, there exists an edge $(i,j)\in E$ independently with probability $\frac{(1-x_i)k}{n}$.
So, in expectation, each individual $i$ owns $k$ books, and $x_i$ fraction of them are \ttr.

For two books $j,j'\in V_2$, let $M_{jj'} := |N(j) \cap N(j')|$ be the number of individuals in $V_1$ that are neighbors of both $j$ and $j'$ in $G$. 
For any two nodes $i,j\in V$, let $\prob[i\xrightarrow{\ell} j]$ be the probability that a $\ell$-step random walk over $G$ starting at $i$ ends at $j$.
For a node $i\in V_1$ and a node $j\in V_2$, let $Z_{ij}$ be the indicator variable for edge $(i,j)$,\ie, $Z_{ij} = 1$ if $(i,j)\in E$, and $Z_{ij} = 0$ otherwise.

\subsection{Analysis}
Next we prove our results about the polarizing effects of each of the three algorithms. Our results hold with probability 1 in the limit as $n\rightarrow\infty$.
First we invoke the Strong Law of Large Numbers to show that the random quantities we care about all take their expected values with probability 1 as $n\rightarrow \infty$.
\begin{lem}\label{lem:ssln}
In the limit as $n\rightarrow\infty$, with probability 1,
\begin{enumerate}[(a)]
\item
for all $i\in V_1,\ |N(i)| \rightarrow k$,
\item 
for all $i\in V_1,\ \sum_{\substack{j_1\in V_2\\j_1 \text{ is \ttr}}} Z_{ij_1} \rightarrow x_ik$,
\item  
for all $i\in V_1,\ \sum_{\substack{j_1\in V_2\\j_2 \text{ is \ttb}}} Z_{ij_2} \rightarrow (1-x_i)k$,
\item
for all $j\in V_2,\ |N(j)| \rightarrow \frac{mk}{2n}$,
\item
for every pair of \ttr\ books $j,j'\in V_2, M_{jj'} = \sum_{i\in V_1} Z_{ij}Z_{ij'} \rightarrow \frac{mk^2(\frac14 + \var(x_1))}{n^2}$,
\item
for every pair of \ttb\ books $j,j'\in V_2, M_{jj'} = \sum_{i\in V_1} Z_{ij}Z_{ij'} \rightarrow \frac{mk^2(\frac14 + \var(x_1))}{n^2}$, and 
\item
for every \ttr\ book $j$ and every \ttb\ book $j',\ M_{jj'} = \sum_{i\in V_1} Z_{ij}Z_{ij'} \rightarrow \frac{mk^2(\frac14 - \var(x_1))}{n^2}$.
\end{enumerate}
\end{lem}
\begin{proof}
Recall that as $n\rightarrow \infty$, $m = f(n)\rightarrow\infty$. So statements (a) through (g) follow from the Strong Law of Large Numbers.
\end{proof}

\begin{algorithm}[t]
\begin{algorithmic}
\STATE {\bf Input:} $G = (V_1, V_2, E)$, node $i\in V_1$.
\STATE {\bf Parameter:} A large positive integer $T$.
\STATE Choose a neighbor $k$ of $i$ uniformly at random.
\STATE Perform $T$ two-step random walks on $G$ starting at $k$.
\STATE For node $j\in V_2$, let \texttt{count(j)} be the number of random walks that end at node $j$.
\STATE {\bf Output:} $j^* := \arg\max_j \texttt{count(j)}$.
\end{algorithmic}
\caption{SimpleICF}
\label{rec:ICF}
\end{algorithm}
We use Lemma~\ref{lem:ssln} to prove our results.
First we show that SimplePPR (Algorithm~\ref{rec:PPR}) is polarizing with respect to $i$ even if  $i$ is unbiased.
\begin{thm}\label{thm:ppr-n->infty}
In the limit as $n\rightarrow\infty$ and as $T\rightarrow\infty$, SimplePPR  is polarizing with respect to $i$.
\end{thm}
Next we show that SimpleSALSA and SimpleICF are polarizing only if $i$ is biased.
\begin{thm}\label{thm:salsa-icf-n->infty}
In the limit as $n\rightarrow\infty$,
\begin{enumerate}
\item
SimpleSALSA  is polarizing with respect to $i$ if and only if $i$ is biased.
\item
In the limit as $T\rightarrow \infty$, SimpleICF is polarizing with respect to $i$ if and only if $i$ is biased.
\end{enumerate}
\end{thm}
Both Theorem~\ref{thm:ppr-n->infty} and Theorem~\ref{thm:salsa-icf-n->infty} are proved in Appendix~\ref{app:proofs-sec5}.

\section{Discussion of Various Measures of Opinion Divergence}\label{sec:polarization-metrics}
Recall that we define an opinion formation process to be polarizing if it results in an increased divergence of opinions.
Here we describe a number of alternate measures of divergence, and discuss how many of our results hold for these measures.
A generalization of the global disagreement index (GDI) is the following: $\sum_{i < j} h(|x_i - x_j|)$, where $h$ is an arbitrary convex function.
The flocking process has the property that the vector $\bx(t+1)$ is majorized by $\bx(t)$.
 Therefore, as noted in the proof of Theorem~\ref{thm:flocking-not-polarizing}, each opinion update of the  flocking process is depolarizing under this definition, or more generally, when divergence is defined by any symmetric convex function of $\bx$.
 
A stronger definition of divergence is one based on second order stochastic dominance, which is defined over distributions, but can be easily modified to work with vectors.
Informally, a distribution $F$ is second order stochastically dominated by a distribution $G$ if $F$ is a mean-preserving spread of $G$.
 Let us say an opinion formation process is polarizing if the final opinion vector is dominated (second order stochastically) by the initial opinion vector, and is depolarizing if the final vector dominates the initial vector.
 According to this definition, a single opinion update in the DeGroot and flocking processes is in general neither polarizing nor depolarizing.
However, both these processes have been shown to converge to consensus under fairly general conditions (\cite{degroot:consensus, jnt-phdthesis}).
Thus, under those conditions, both these processes are depolarizing in equilibrium. 
Moreover, our results on the three recommender algorithms (Theorem~\ref{thm:ppr-n->infty} and Theorem~\ref{thm:salsa-icf-n->infty}) also hold under this definition of divergence.

Consider the following even stronger definition of polarization: a process is polarizing if at each time step, it pushes the opinions of individuals away from the average and is  depolarizing if  it brings their opinions closer to the average. 
Under this definition too, the DeGroot and flocking processes are neither polarizing nor depolarizing.
However, under all three definitions, the biased opinion formation process is polarizing on a two-island network when $b \ge 1$.

\section{Conclusion}
In this paper we attempted to explain polarization in society through a model of opinion formation.
We generalized DeGroot's repeated averaging model to account for biased assimilation.
We showed that DeGroot-like repeated averaging processes can never be polarizing, even if individuals are arbitrarily homophilous. 
We also showed that in a two-island network, our biased opinion formation process may result in either polarization (if $b\ge 1$), persistent disagreement (if $1 > b \ge 2/(h+1)$), or consensus (if $b < 2/(h+1)$).
In other words, homophily alone, without biased assimilation, is not sufficient to polarize society.
We used biased assimilation to provide insight into the polarizing effects of three recommender algorithms: SimpleSALSA, SimplePPR and SimpleICF.
We showed that for a simple, natural model of the underlying user-item graph, SimpleSALSA and SimpleICF are polarizing only if individuals are biased whereas SimplePPR is polarizing even if  individuals are unbiased.

One direction for further investigation is to study through human subject experiments how the degree of homophily and the strength of biased assimilation affect whether individuals interacting over a network polarize or arrive at a consensus?
Our analysis of recommender algorithms is a first step toward designing algorithms and online social systems that counteract polarization and facilitate greater consensus between individuals over complex and vexing social, economic and political issues.
We view this as a promising and important direction for further research. 
\bibliographystyle{alphaurl}
\bibliography{polarization}

\appendix

\section{Proof of Theorem~\ref{thm:single-agent}}\label{app:proofs-sec2}
Recall that
\[
x(t+1) :=  \frac{wx(t) + (x(t))^bs}{w + (x(t))^bs + (1-x(t))^b(1-s)}
\]
Equivalently,
\beq\label{eq:thm1.1}
\frac{x(t+1)}{1-x(t+1)} = \frac{wx(t) + (x(t))^bs}{w(1-x(t)) + (1-x(t))^b(1-s)} = \frac{w + (x(t))^{b-1}s}{w + (1-x(t))^{b-1}(1-s)} \frac{x(t)}{1-x(t)}
\eeq
First we will show that if $x(t) = \hx$, then for all $t'> t, x(t') = \hx$.
\begin{lem}
Assume $b\ne 1$. Fix $t\ge 0$. Let $x(t) = \hx$. Then for all $t'> t, x(t') = \hx$.
\end{lem}
\begin{proof}
To prove the lemma, it suffices to show that $x(t+1) = x(t) = \hx$.
Recall that
\[
\hx := \frac{s^{1/(1-b)}}{s^{1/(1-b)} + (1-s)^{1/(1-b)}}
\]
Or equivalently,
\[
\left(\frac{\hx}{1-\hx}\right)^{1-b} = \frac{s}{1-s} 
\]
This implies that when $x(t) = \hx$, $x(t)^{b-1}s = (1-x(t))^{b-1}(1-s)$. Substituting this in \eqref{eq:thm1.1}, we get that
\[
\frac{x(t+1)}{1-x(t+1)} = \frac{x(t)}{1-x(t)}
\]
Or equivalently, $x(t+1) = x(t)$.
\end{proof}

Next we will show that when $b > 1$, $\hx$ is an unstable equilibrium.
\begin{lem}\label{lem:thm1-monotone}
Let $b > 1$. Fix $t \ge 0$. 
\begin{enumerate}
\item
If $x(t) > \hx$, then $x(t+1) > x(t)$.
\item
If $x(t) < \hx$, then $x(t+1) < x(t)$.
\end{enumerate}
\end{lem}
\begin{proof}
Again, recall that
\[
\left(\frac{\hx}{1-\hx}\right)^{1-b} = \frac{s}{1-s} 
\]
Therefore, if $x(t) > \hx$, it implies that
\[
\frac{x(t)}{1-x(t)} > \frac{\hx}{1-\hx}  \Rightarrow \left(\frac{x(t)}{1-x(t)}\right)^{1-b} < \left(\frac{\hx}{1-\hx}\right)^{1-b}  = \frac{s}{1-s} \text{ (since $b > 1$)}
\]
Or equivalently, $(x(t))^{b-1}s > (1-x(t))^{b-1}(1-s)$. Substituting this in \eqref{eq:thm1.1}, we get that
\[
\frac{x(t+1)}{1-x(t+1)} > \frac{x(t)}{1-x(t)}
\]
Or equivalently, $x(t+1) > x(t)$.

By a similar argument, if $x(t) < \hx$, then$(x(t))^{b-1}s < (1-x(t))^{b-1}(1-s)$. Again, substituting this in \eqref{eq:thm1.1}, we get that
\[
\frac{x(t+1)}{1-x(t+1)} < \frac{x(t)}{1-x(t)}
\]
Or equivalently, $x(t+1) < x(t)$.
\end{proof}
Next we will show that when $b > 1$, either $\lim_{t\rightarrow\infty} x(t) = 1$ or $\lim_{t\rightarrow\infty} x(t) = 0$.
\begin{lem}
Let $b > 1$. Fix $t \ge 0$. 
\begin{enumerate}
\item
If $x(t) > \hx$, then $\lim_{t\rightarrow\infty} x(t) = 1$.
\item
If $x(t) < \hx$, then $\lim_{t\rightarrow\infty} x(t) = 0$.
\end{enumerate}
\end{lem}
\begin{proof}
For the proof, we will assume that $x(t) > \hx$ and show that $\lim_{t\rightarrow\infty} x(t) = 1$. The case when $x(t) < \hx$ can be argued in an analogous way.

By definition, we know that for all $t\ge 0, x(t) \in [0,1]$.
Further, from Lemma~\ref{lem:thm1-monotone}, we know that the sequence $\{x(t')_{t'\ge t}\}$ is strictly increasing.
Since the sequence is strictly increasing and bounded, it must converge either to 1 or to some value in the interval $[x(t), 1)$.
Consider the function $g: [0,1]\rightarrow \real$ defined as
\[
g(y) := \frac{w + y^bs}{w + y^bs + (1-y)^b(1-s)} - y
\]
Observe that for all $t\ge 0,\ x(t+1) - x(t) = g(x(t))$. 
Therefore, 
\begin{enumerate}[(a)]
\item
for all $y\in [x(t),1),\ g(y) > 0$ (since, by Lemma~\ref{lem:thm1-monotone}, the sequence $\{x(t')_{t'\rightarrow t}\}$ is strictly increasing), and
\item
$g(1) = 0$.
\end{enumerate}
For the purpose of contradiction, assume that $\lim_{t\rightarrow \infty} x(t) = a$, where $x(t) \le a < 1$. 
This implies, for every $\e > 0$, there exists a $t(\e)$ such that for all $t' \ge t(\e),\ x(t'+1) - x(t') < \e$, or equivalently, that for all $t' \ge t(\e),\ g(x(t')) < \e$.

Let $\min_{y\in[x(t),a]} g(y) = c$. It implies for all $y\in [x(t), a],\ g(y) \ge c$. From (a), it follows that $c > 0$. Setting $\e = c$, our analysis implies the following two properties of $g$: (1) for all $t\ge 0, g(x(t)) \ge c$, and (2) for all $t' \ge t(\e), g(x(t')) < c$, which contradict each other. This completes the proof by contradiction.
\end{proof}
Using a similar argument we can show that when $b < 1$, $\hx$ is a stable equilibrium.
\begin{lem}
Let $b < 1$. Fix $t \ge 0$. 
\begin{enumerate}
\item
If $x(t) > \hx$, then $x(t+1) < x(t)$.
\item
If $x(t) < \hx$, then $x(t+1) > x(t)$.
\end{enumerate}
\end{lem}
\begin{lem}
Let $b < 1$. Then, $\lim_{t\rightarrow\infty} x(t) = \hx$.
\end{lem}
\section{Proofs of Section 3}\label{app:proofs-sec3}

\begin{proof}[Proof of Theorem~\ref{thm:degroot-not-polarizing}]
Recall that since $b_i = 0$, the opinion of node $i$ at time $t+1$ is given by
\begin{equation}\label{eq:5.1}
x_i(t+1) = \frac{w_{ii}x_i(t) + \sum_{j\in N(i)} w_{ij}x_j(t)}{w_{ii}+d_i}
\end{equation}
where recall that $d_i := \sum_{j\in N(i)} w_{ij}$ is the weighted degree of node $i$.
Let $L_G$ be the weighted laplacian matrix of $G$. Recall that $L_G$ is given by
\[
(L_G)_{ij} = \left\{
\begin{array}{rl}
d_i,&\ \text{ if } i = j\\
-w_{ij},&\ \text{ if } (i,j)\in E\\
0,&\ \text{ otherwise}
\end{array}
\right.
\]
Now consider the vector $L_G\bx(t)$. The $i$th entry of the vector is given by
\begin{align*}
(L_G\bx(t))_i = d_ix_i(t) - \sum_{j\in N(i)} w_{ij}x_j(t) &= d_ix_i(t) + w_{ii}x_i(t) - \left(w_{ii}x_i(t) + \sum_{j\in N(i)} w_{ij}x_j(t)\right)\\
&= (d_i + w_{ii})(x_i(t) - x_i(t+1)) \text{ (from \eqref{eq:5.1})}
\end{align*}
Equivalently, in matrix notation,
\begin{equation}\label{eq:5.2}
\bx(t+1) = (I - DL_G)\bx(t)
\end{equation}
where, $D$ is a diagonal matrix such that $D_{ii} = 1/(d_i + w_{ii})$. Note that since $G$ is connected, $d_i > 0$, and therefore $D_{ii}$ is finite.
Consider the difference $\eta(G, \bx(t+1)) - \eta(G, \bx(t))$. Observe that for a vector $\by\in [0,1]^n,\ \eta(G, \by) = \by^{\top}L_G\by$. Therefore, we have that
\begin{align*}
\eta(G, \bx(t+1)) - \eta(G, \bx(t)) &= (\bx(t+1))^{\top}	L_G(\bx(t+1)) - (\bx(t))^{\top}L_G\bx(t) \\
&= (\bx(t))^{\top}(I - DL_G)^{\top}L_G(I - DL_G)\bx(t) - (\bx(t))^{\top}L_G\bx(t) \text{ (from \eqref{eq:5.2})}\\
&= (\bx(t))^{\top}\left((L_G - L_GDL_G)(I - DL_G) - L_G\right) \bx(t) \text{ (since $L_G$ is symmetric)}\\
&= (\bx(t))^{\top}\left(L_G - L_GDL_G - L_GDL_G - L_GDL_GDL_G - L_G\right) \bx(t)\\
&= (\bx(t))^{\top}\left(L_GDL_GDL_G - 2L_GDL_G \right) \bx(t)\\
&= (\bx(t))^{\top}L^{\top}_GD^{1/2}((D^{1/2}L_GD^{1/2} - 2I))D^{1/2}L_G\bx(t) \text{ (since $L_G$ is symmetric)}\\
&= \by^{\top}(D^{1/2}L_GD^{1/2} - 2I)\by \text{ (where $\by := D^{1/2}L_G\bx(t)$)}
\end{align*}
Thus, in order to show that $\eta(G, \bx(t+1)) - \eta(G, \bx(t)) \le 0$, it suffices to show that for all vectors $\by\in \real^n,\ \by^{\top}D^{1/2}L_GD^{1/2}\by \le 2||\by||^2_2$. We prove this as Lemma~\ref{lem:lin-algebra}.
\end{proof}

\begin{lem}\label{lem:lin-algebra}
Consider an arbitrary weighted undirected graph $G = (V,E,w)$ over $n$ nodes. Let $L_G$ be the weighted laplacian matrix of $G$. Let $D$ be an $n\times n$ diagonal matrix such that for $i=1,\dotsc, n,\ D_{ii} = 1/(d_i + w_{ii})$, where $d_i = \sum_{j\in N(i)} w_{ij}$ is the weighted degree of $i$ in $G$. Let $\by\in \real^n$ be an arbitrary vector. Then, $\by^{\top}D^{1/2}L_GD^{1/2}\by \le 2||\by||^2_2$.
\end{lem}
\begin{proof}
For $i = 1, \dotsc, n$, let $r_i := d_i + w_{ii}$.
Let $P := D^{1/2}L_GD^{1/2}$. Then,
\[
P_{ij} = \left\{
\begin{array}{rl}
\frac{d_i}{r_i},&\ i = j\\
\frac{-w_{ij}}{\sqrt{r_ir_j}},&\ (i,j)\in E\\
0,&\ \text{ otherwise}
\end{array}
\right.
\]
Then, we have that
\begin{align*}
\by^{\top}P\by = \sum_{i,j} P_{ij}y_iy_j &= \sum_{i=1}^n P_{ii}y_i^2 + 2\sum_{(i,j)\in E} P_{ij}y_iy_j = \sum_i \frac{d_i}{r_i}y_i^2 - 2\sum_{(i,j)\in E} \frac{w_{ij}}{\sqrt{r_ir_j}}y_iy_j\\
&= \sum_i\left(\frac1{r_i}y_i^2\sum_{j\in N(i)} w_{ij}\right)  - 2\sum_{(i,j)\in E} \frac{w_{ij}}{\sqrt{r_ir_j}}y_iy_j\\
&= \sum_{(i,j)\in E} w_{ij}\left(\frac{y_i^2}{r_i} + \frac{y_j^2}{r_j}\right)- 2\sum_{(i,j)\in E} \frac{w_{ij}}{\sqrt{r_ir_j}}y_iy_j\\
&= \sum_{(i,j)\in E} w_{ij}\left(\frac{y_i}{\sqrt{r_i}} - \frac{y_j}{\sqrt{r_j}}\right)^2\\
&= -\sum_{(i,j)\in E} w_{ij}\left(\frac{y_i}{\sqrt{r_i}} + \frac{y_j}{\sqrt{r_j}}\right)^2 + 2\sum_i \frac{d_i}{r_i}y_i^2\\
&\le -\sum_{(i,j)\in E} w_{ij}\left(\frac{y_i}{\sqrt{r_i}} + \frac{y_j}{\sqrt{r_j}}\right)^2 + 2\sum_i y_i^2 \text{ (since $d_i \le r_i$)}\\
&\le 2||\by||_2^2
\end{align*}
\end{proof}

\begin{proof}[Proof of Theorem~\ref{thm:flocking-not-polarizing}]
Let $|S(t)| = k$. Then, the opinion update \eqref{def:flocking-update} under the flocking process can be written in matrix form as
\[
\bx(t+1) = (1-\e)\bx(t) + \e A(t)\bx(t)
\]
where $A(t)$ is a $n\times n$ matrix given by
\[
A_{ij}(t) = \left\{
\begin{array}{rl}
\frac1k,&\ \text{ if } i\in S(t), j\in S(t)\\
1,&\ \text{ if } i = j\text{ and } i\notin S(t)\\
0,&\ \text{ otherwise}
\end{array}	
\right.
\]
Observe that $A(t)$ is doubly-stochastic. Then
\begin{align*}
\g(\bx(t+1)) &= \g((1-\e)\bx(t) + \e A(t)\bx(t)) \text{ (by definition of $\bx(t+1)$)}\\
&\le (1-\e)\g(\bx(t)) + \e\g(A(t)\bx(t)) \text{ (since $\g$ is convex in $\bx$)}\\
&\le  (1-\e)\g(\bx(t)) + \e\g(\bx(t)) \text{ (by Proposition~\ref{prop:schur-convex})}\\
&= \g(\bx(t))
\end{align*}

\begin{prop}\label{prop:schur-convex}
$\g(A(t)\bx(t)) \le \g(\bx(t))$.
\end{prop}
\begin{proof}
Let $\by := A(t)\bx(t)$. Since $A(t)$ is doubly stochastic, it follows by a famous theorem by Hardy, Littlewood and Polya, that $\bx(t)$ majorizes $\by$. Moreover, $\g(\bx)$ is a convex symmetric function. Therefore, it is a Schur-convex function. By definition, a function $f: \real^n\rightarrow \real$ is Schur-convex if $f(\bx_1) \ge f(\bx_2)$ whenever $\bx_1$ majorizes $\bx_2$.   Therefore, $\g(\by) \le \g(\bx(t))$.
\end{proof}
\end{proof}

\section{Proof of Theorem~\ref{thm:homophily-polarization}}\label{app:proofs-sec4}
To prove the theorem, we begin by making three simple observations that hold for all $b\ge 0$.
The first observation follows directly from the symmetry of nodes in each set $V_1$ and $V_2$.
\begin{lem}\label{lem:identical-expectation}
Consider nodes $i, j\in V$ such that either both $i, j\in V_1$ or both $i, j\in V_2$. Then for all $t\ge 0,\ x_i(t) = x_j(t)$.
\end{lem}

The next observation allows us to focus on only analyzing the equilibrium opinion of nodes in $V_1$.
\begin{lem}\label{lem:comp-expectation}
Consider a node $i\in V_1$ and a node $j\in V_2$. Then, for all $t \ge 0,\ x_i(t) = 1-x_j(t)$.
\end{lem}
\begin{proof}[Proof of Lemma~\ref{lem:comp-expectation}]
By induction. 

Induction hypothesis: Assume that the statement holds for some $t \ge 0$. 

Base case: The statement holds for $t = 0$ by assumption in the theorem statement.

We will now show that the statement holds for $t+1$.
\beq\label{eq:islands-1}
\frac{x_i(t+1)}{1-x_i(t+1)} = \frac{(x_i(t))^b}{(1-x_i(t))^b}\frac{s_i(t)}{d_i - s_i(t)}
\eeq
where $d_i = n(p_s + p_d)$ and, by Lemma~\ref{lem:identical-expectation}, $s_i(t) = n(p_sx_i(t) + p_dx_j(t))$. On the other hand,
\beq\label{eq:islands-2}
\frac{x_j(t+1)}{1-x_j(t+1)} = \frac{(x_j(t))^b}{(1-x_j(t))^b}\frac{s_j(t)}{d_j - s_j(t)}
\eeq
where $s_j(t) = n(p_sx_j(t) + p_dx_i(t))$, and $d_j = n(p_s + p_d) = d_i$. 
By the induction hypothesis, we know that $x_i(t) = 1- x_j(t)$. It follows that $S_i(t) = d_i - s_j(t)$. Substituting this into \eqref{eq:islands-1}, we get
\[
\frac{x_i(t+1)}{1-x_i(t+1)} = \frac{(x_i(t))^b}{(1-x_i(t))^b}\frac{s_i(t)}{d_i - s_i(t)} = \frac{(1 - x_j(t))^b}{(x_j(t))^b}\frac{d_j - s_j(t)}{s_j(t)} = \frac{1 - x_j(t+1)}{x_j(t+1)} 
\]
where the last equality follows from \eqref{eq:islands-2}. It follows that $x_i(t+1) = 1-x_j(t+1)$.

This completes the inductive proof.
\end{proof}

Lemma~\ref{lem:comp-expectation} implies that if we prove the theorem statement for nodes in $V_1$, we get the proof for nodes in $V_2$ for free.
So, in the rest of the proof, we only make statements about nodes in $V_1$.
The third observation lower bounds the opinions of nodes in $V_1$.

\begin{lem}\label{lem:x-in-[0,1]}
Consider a node $i\in V_1$. For all $t\ge 0,\ x_i(t) \in [\frac12, 1]$.
\end{lem}
\begin{proof}[Proof of Lemma~\ref{lem:x-in-[0,1]}]
It is easy to see that for all $t\ge 0,\ x_i(t) \le 1$. We will prove that $x_i(t) \ge \frac12$ by induction over $t$.

Base case: The statement holds for $t = 0$ by assumption in the theorem statement.

Induction hypothesis: Assume that the lemma  statement holds for some $t \ge 0$, \ie, assume that $x_i(t) \ge \frac12$ for some $t  \ge 0$.

We will show that the lemma statement holds for $t+1$.
\begin{align*}
\frac{x_i(t+1)}{1 - x_i(t+1)} &= \frac{(x_i(t))^b}{(1-x_i(t))^b}\frac{S_i(t)}{d_i - s_i(t)}\\
&\ge \frac{(x_i(t))^b}{(1-x_i(t))^b} \text{ (since $s_i(t) > d_i - s_i(t)$)}\\
&\ge 1 \text{ (since $x_i(t) \ge \frac12$ by the induction hypothesis, and $b\ge 0$)}
\end{align*}
This implies $x_i(t+1) \ge \frac12$, completing the inductive proof.
\end{proof}

Recall that $i$'s opinion at time $t+1$ is given by
\[
x_i(t+1) = \frac{(x_i(t))^bs_i(t)}{(x_i(t))^bs_i(t) + (1-x_i(t))^b(d_i - s_i(t))} \text{ (by \eqref{def:biased-update})}
\] 
where $s_i(t) = n(p_sx_i(t) + p_d(1-x_i(t)))$, and $d_i = n(p_s + p_d)$.
Now consider the equation
\beq\label{eq:x(t)=x(t+1)}
x_i(t+1) = x_i(t)
\eeq
We will show that if $b\ge 1$ or $b < \frac2{h_G+1}$, \eqref{eq:x(t)=x(t+1)} has no solution in $(\frac12, 1)$, whereas if $1 > b \ge \frac2{h_G+1}$, there exists a unique solution to \eqref{eq:x(t)=x(t+1)} in $(\frac12, 1)$.
\begin{lem}\label{lem:x(t+1)x(t)}
Consider a node $i\in V_1$. Fix $t\ge 0$. 
\begin{enumerate}[(a)]
\item
If $b\ge 1$,  for every $x_i(t)\in (\frac12,1),\ x_i(t+1) > x_i(t)$.
\item
If $1 > b \ge \frac2{h_G+1}$, there exists a unique solution, say $\hat{x}$, to Eq.\eqref{eq:x(t)=x(t+1)} in  $(\frac12,1)$.
\item
If $b < \frac2{h_G+1}$, for every $x_i(t)\in (\frac12,1),\ x_i(t+1) < x_i(t)$.
\end{enumerate}
\end{lem}
\begin{proof}[Proof of Lemma~\ref{lem:x(t+1)x(t)}]
Consider the function $f: [0,1]\rightarrow \real$ defined as
\beq\label{def:f}
f(y; b) := \left\{
\begin{array}{rl} 
1,&\ y\in [0,1], b = 1\\
0, &\ y\in [0,1], b = 2\\
\frac2{b} - 1,&\ y = \frac12, b > 0\\
\frac{(y)^{2-b} - (1-y)^{2-b}}{y(1-y)^{1-b} - y^{1-b}(1-y)},&\ \text{ otherwise} 
\end{array}
\right.
\eeq

We will first prove a few properties of $f$ and then use those properties to prove Lemma~\ref{lem:x(t+1)x(t)}.

\begin{prop}\label{prop:f-1/2}
\begin{enumerate}
\item
For all $b > 0$, $f$ is continuous over $[0,1]$.
\item
If $0 < b <  1$, $f$ is strictly increasing over $[\frac12,1]$.
\item
If $b \ge 1$, for all $y\in [0,1), f(y; b) \le 1$.
\end{enumerate}
\end{prop}
\begin{proof}
\begin{enumerate}
\item
Observe that $f$ is continuous when $b=1$ or $b=2$.
So, we only need to show that $f$ is continuous at $y = \frac12$ when $b\ne 1$ and $b\ne 2$. 
Let $p(y; b) := (y)^{2-b} - (1-y)^{2-b}$ and $q(y; b) := y(1-y)^{1-b} - y^{1-b}(1-y)$. Observe that when $b\ne 1$ and $b\ne 2$,  both $p$ and $q$ are differentiable on $[0,1]$. For $y\in [0,1]$,
\[
p'(y; b) = (2-b)(y^{1-b} + (1-y)^{1-b}); q'(y; b) = (1-y)^{1-b} - (1-b)y(1-y)^{-b} - (1-b)y^{-b}(1-y) + y^{1-b}
\]
Therefore,
\beq\label{eq:propf12.1.1}
\lim_{y\rightarrow 1/2}\frac{p'(y; b)}{q'(y; b)} = \lim_{y\rightarrow 1/2} \frac{(2-b)(y^{1-b} + (1-y)^{1-b})}{(1-y)^{1-b} - (1-b)y(1-y)^{-b} - (1-b)y^{-b}(1-y) + y^{1-b}} = \frac2{b} - 1
\eeq
So, we have that
\begin{align*}
\lim_{y \rightarrow 1/2}f\left(y; b\right) &= \lim_{y\rightarrow 1/2} \frac{p(y; b)}{q(y; b)} = \lim_{y\rightarrow 1/2} \frac{p'(y)}{q'(y)} \text{ (using L'H\^{o}pital's rule)} = \frac2b - 1 \text{ (from \eqref{eq:propf12.1.1}) } = f(\frac12; b)
\end{align*}
Therefore, when $b\ne 1$ and $b\ne 2$, $f$ is continuous at $\frac12$. 
\item
Assume $0 < b < 1$. Fix $y_1, y_2\in [\frac12,1]$ such that $y_1 > y_2$. We will show that $f(y_1; b) > f(y_2); b$.
For conciseness of expression, define $\bar{y_1} := 1-y_1$ and $\bar{y_2} := 1 - y_2$.
Then
\beq\label{eq:propf12.1}
y_1y_2 - y_1\bar{y_2}  > (y_1y_2)^{1-b} - (y_1\bar{y_2})^{1-b}
\eeq
Similarly,
\beq\label{eq:propf12.2}
\bar{y_1}y_2 - \bar{y_1}\bar{y_2}   > (\bar{y_1} y_2)^{1-b} - (\bar{y_1} \bar{y_2} )^{1-b}
\eeq
Adding \eqref{eq:propf12.1} and \eqref{eq:propf12.2}, we get
\[
y_1y_2 - y_1\bar{y_2}  + \bar{y_1}y_2 - \bar{y_1} \bar{y_2} > (y_1y_2)^{1-b} - (y_1\bar{y_2})^{1-b} + (\bar{y_1} y_2)^{1-b} - (\bar{y_1} \bar{y_2} )^{1-b}
\]
Or equivalently,
\beq\label{eq:propf12.3}
\left(y_1y_2 - \bar{y_1} \bar{y_2}\right) - \left((y_1y_2)^{1-b}  - (\bar{y_1}\bar{y_2})^{1-b}\right) > \left(y_1\bar{y_2}  - \bar{y_1} y_2\right) - \left((y_1\bar{y_2} )^{1-b} - (\bar{y_1} y_2)^{1-b}\right)
\eeq
Moreover, since $y_1, y_2\in [\frac12,1]$ and $y_1 > y_2$,
\beq\label{eq:propf12.4}
y_1y_2 - \bar{y_1} \bar{y_2} > 0; (y_1y_2)^{1-b}  - (\bar{y_1}\bar{y_2})^{1-b} > 0; y_1\bar{y_2}  - \bar{y_1} y_2 > 0; (y_1\bar{y_2} )^{1-b} - (\bar{y_1} y_2)^{1-b} > 0
\eeq
\eqref{eq:propf12.3} and \eqref{eq:propf12.4} imply that
\[
\frac{y_1y_2 - \bar{y_1}\bar{y_2}}{y_1\bar{y_2} - \bar{y_1}y_2} > \frac{(y_1y_2)^{1-b}  - (\bar{y_1}\bar{y_2})^{1-b}}{(y_1\bar{y_2})^{1-b} - (\bar{y_1}y_2)^{1-b}}
\]
Rearranging, we get
\[
\frac{(y_1)^{2-b} - \bar{y_1}^{2-b}}{y_1\bar{y_1}^{1-b} - y_1^{1-b}\bar{y_1}} = f(y_1; b) > \frac{(y_2)^{2-b} - \bar{y_2}^{2-b}}{y_2\bar{y_2}^{1-b} - y_2^{1-b}\bar{y_2}} = f(y_2;b)
\]
\item
Since $f$ is symmetric about $y = \frac12$, we will prove the theorem for $y\in [\frac12,1)$. Fix $y\in [\frac12,1)$. Observe that when $b\ge 1$, $(1-y)^{1-b} \ge y^{1-b}$ (since $y \ge 1-y$). Equivalently
\beq\label{eq:lem4.5.1}
y(1-y)^{1-b} \ge y^{2-b}
\eeq
For the same reason,
\beq\label{eq:lem4.5.2}
y^{1-b}(1-y) \le (1-y)^{2-b}
\eeq
From \eqref{eq:lem4.5.1} and \eqref{eq:lem4.5.2}, it follows that
\[
y(1-y)^{1-b} - y^{1-b}(1-y) \ge (y)^{2-b} - (1-y)^{2-b}
\]
or equivalently, $f(y; b) \le 1$.
\end{enumerate}
\end{proof}

Using these properties of $f$ we will prove Lemma~\ref{lem:x(t+1)x(t)}.
\begin{enumerate}
\item
If $b \ge 1$, then for all $y\in [0,1),\ f(y; b) \le 1 \text{ (by Proposition~\ref{prop:f-1/2})}  < h_G$. Therefore, for $y\in [\frac12, 1)$,
\begin{align*}
 \frac{(y)^{2-b} - (1-y)^{2-b}}{y(1-y)^{1-b} - y^{1-b}(1-y)}  &< h_G\\
 \Leftrightarrow y^{2-b}  - (1-y)^{2-b} &< h_G(y(1-y)^{1-b} - y^{1-b}(1-y))\\
 \Leftrightarrow y^{2-b} + h_Gy^{1-b}(1-y) &< (1-y)^{2-b} + h_Gy(1-y)^{1-b}\\
 \Leftrightarrow y^{1-b} (y + (1-y)h_G) &< (1-y)^{1-b}((1-y) + h_Gy)\\
 \Leftrightarrow \frac{y}{1-y} &< \left(\frac{y}{1-y}\right)^b\cdot\frac{(1-y) + h_Gy}{y + (1-y)h_G}
 \end{align*}
For $y = x_i(t)$, the right hand side of the last inequality above is equal to $x_i(t+1)/(1-x_i(t+1))$, implying that $x_i(t+1) > x_i(t)$.
\item
If $1 > b \ge \frac2{h_G+1}$, then observe that $f(\frac12; b) = \frac2{b} - 1 \le h_G < f(1;b) = \infty$. Since $f$ is a continuous function (by Proposition~\ref{prop:f-1/2}), therefore, by the intermediate value theorem, there must exist a $\hat{y}\in [\frac12,1)$ such that $f(\hat{y}; b) = h_G$.
Equivalently,
\[
 \frac{(\hat{y})^{2-b} - (1-\hat{y})^{2-b}}{\hat{y}(1-\hat{y})^{1-b} - \hat{y}^{1-b}(1-\hat{y})}  = h_G
\]
Rearranging the above expression, we get
\[
\frac{\hat{y}}{1-\hat{y}} = \left(\frac{\hat{y}}{1-\hat{y}}\right)^b\cdot\frac{(1-\hat{y}) + h_G\hat{y}}{\hat{y} + (1-\hat{y})h_G}
\]
Again, for $\hat{y} = x_i(t)$, we have that $ x_i(t+1) = x_i(t)$. The uniqueness of $\hat{x}$ follows from the fact that, by Proposition~\ref{prop:f-1/2}, $f$ is strictly increasing over $(\frac12,1]$.
\item
If $b < \frac2{h_G+1}$, then for all $y\in [\frac12,1],\ f(y;b) \ge f(\frac12; b) \text{ (by Proposition~\ref{prop:f-1/2})} = \frac2{b} - 1 > h_G$.
In other words,
\[
 \frac{(y)^{2-b} - (1-y)^{2-b}}{y(1-y)^{1-b} - y^{1-b}(1-y)}  > h_G
 \]
Again, rearranging the above expression, we get
\[
\frac{y}{1-y} > \left(\frac{y}{1-y}\right)^b\cdot\frac{(1-y) + h_Gy}{y + (1-y)h_G}\\
\] 
Again, for $y = x_i(t) $, the right hand side of the last inequality above is equal to $x_i(t+1)$, implying that $x_i(t+1) < x_i(t)$.
\end{enumerate}
This concludes the proof of Lemma~\ref{lem:x(t+1)x(t)}.
\end{proof}

Next we will prove Theorem~\ref{thm:homophily-polarization} for the case of persistent disagreement, the cases of polarization and consensus are limiting cases of that case as $b\rightarrow 1$ and $b\rightarrow 2/(h_G+1)$ respectively.
We will show that when $1 > b \ge \frac2{h_G+1}$, the value $\hx$ defined in Lemma~\ref{lem:x(t+1)x(t)}(b) is a stable equilibrium.
The other two cases can be formally proven using an argument similar to the one below.
Next we will show that when $1 > b \ge \frac2{h_G+1}$, the sequence $\{x_i(t)\}$ is bounded.

\begin{lem}\label{lem:bounded-opinion}
Consider a node $i\in V_1$. Let $1 > b \ge  \frac2{h_G+1}$. Let $\hat{x}\in (\frac12,1)$ be the solution to \eqref{eq:x(t)=x(t+1)}.
\begin{enumerate}
\item
If $x_0 < \hx$, then for all $t> 0,\ x_i(t) < \hx$.
\item
If $x_0 > \hx$, then for all $t> 0,\ x_i(t) > \hx$.
\end{enumerate}
\end{lem}
\begin{proof}[Proof of Lemma~\ref{lem:bounded-opinion}]
We will prove statement (1). Statement (2) can be proven using a similar argument.

Proof by induction.

Induction hypothesis: Assume that the lemma  statement holds for some $t \ge 0$, \ie, assume that $x_i(t) < \hx$ for some $t  \ge 0$.

Base case: The statement holds for $t = 0$ by assumption.

We will show that the lemma statement holds for $t+1$.
\[
\frac{x_i(t+1)}{1 - x_i(t+1)} = \frac{(x_i(t))^b}{(1-x_i(t))^b}\frac{s_i(t)}{d_i - s_i(t)} < \frac{(\hx)^b}{(1-\hx)^b}\frac{s_i(t)}{d_i - s_i(t)} \text{ (since $\frac12 < x_i(t) < \hx$, and $b > 0$)}
\]
Observe that since $x_i(t) < \hx$ and $p_s > p_d,\ s_i(t) =  n(p_sx_i(t) + p_d(1-x_i(t))) <  n(p_s\hx + p_d(1-\hx))$. Therefore, 
\[
\frac{s_i(t)}{d_i - s_i(t)} < \frac{p_s\hx + p_d(1-\hx)}{p_s(1-\hx) + p_d\hx}
\]
As a result,
\[
\frac{x_i(t+1)}{1 - x_i(t+1)} < \frac{(\hx)^b}{(1-\hx)^b}\frac{p_s\hx + p_d(1-\hx)}{p_s(1-\hx) + p_d\hx} = \frac{\hx}{1-\hx} \text{ (by definition of $\hx$)}
\]
This implies $x_i(t+1) < \hx$. 
This completes the inductive proof.
\end{proof}

Next we will prove that when $1 > b \ge \frac2{h_G+1}$, the sequence $\{x_i(t)\}$ is monotone.

\begin{lem}\label{lem:monotone-opinion}
Consider a node $i\in V_1$. Let $1 > b \ge \frac2{h_G+1}$. Let $\hat{x}\in (\frac12,1)$ be the solution to \eqref{eq:x(t)=x(t+1)}.
\begin{enumerate}
\item
If $x_0 < \hx$, the sequence $\{x_i(t)\}$ is strictly increasing.
\item
If $x_0 > \hx$, the sequence $\{x_i(t)\}$ is strictly decreasing.
\end{enumerate}
\end{lem}
\begin{proof}[Proof of Lemma~\ref{lem:monotone-opinion}]
We will prove statement (1); statement (2) can be proven using a similar argument.

Assume $x_0 < \hx$. Then, from Lemma~\ref{lem:bounded-opinion}, we know that for all $t\ge 0, x_i(t) < \hx$. 
Fix $t\ge 0$. 
Let $x_i(t) = y < \hx$.
Recall that by definition of $\hx$, if $x_i(t) = \hx,\ x_i(t+1) = x_i(t)$. Equivalently, $f(\hx;b) = h_G$, where $f$ is defined by \eqref{def:f}. From Proposition~\ref{prop:f-1/2}, we know that $f$ is strictly increasing over the interval $(\frac12,\hx)$. Therefore, $f(y;b) < f(\hx;b) = h_G$. Equivalently,
\[
\frac{(y)^{2-b} - (1-y)^{2-b}}{y(1-y)^{1-b} - y^{1-b}(1-y)} < h_G
\]
Rearranging, we get
\[
\frac{y}{1-y} < \left(\frac{y}{1-y}\right)^b\cdot\frac{(1-y) + h_Gy}{y + (1-y)h_G} = \frac{x_i(t+1)}{1-x_i(t+1)}
\]
Equivalently, $x_i(t+1) > x_i(t)$.
\end{proof}

Using the fact that the sequence $\{x_i(t)\}$ is monotone and bounded, next we will prove that it converges to $\hx$.
\begin{lem}\label{lem:persistent-disagreement}
Consider a node $i\in V_1$. Let $1 > b \ge \frac2{h_G+1}$. Let $\hat{x}\in (\frac12,1)$ be the solution to \eqref{eq:x(t)=x(t+1)}. Then, $\lim_{t\rightarrow\infty} x_i(t) = \hx$.
\end{lem}
\begin{proof}
For the proof, we will assume that the initial opinion $x_i(0) = x_0 \le \hx$. The case when $x_0 > \hx$ can be argued in an analogous way.

Observe that if $x_0 = \hat{x}$, then by Lemma~\ref{lem:x(t+1)x(t)}, it follows that for all $t\ge 0,\ x_i(t+1) = \hat{x}$, and we are done. 
So let us assume that $\frac12 < x_0 < \hat{x}$.
From Lemma~\ref{lem:bounded-opinion} and Lemma~\ref{lem:monotone-opinion}, we know that the sequence $\{x_i(t)\}$ is strictly increasing and bounded. 
This implies that the sequence must converge either to $\hx$ or to some value in the interval $[x_0, \hx)$.
Consider the function $g: [0,1]\rightarrow \real$ defined as
\[
g(y) := \frac{y^b(h_Gy + (1-y))}{y^b(h_Gy + (1-y) + (1-y)^b(h_G(1-y) + y)} - y
\]
Observe that for all $t\ge 0,\ x_i(t+1) - x_i(t) = g(x_i(t))$. 
Therefore, 
\begin{enumerate}[(a)]
\item
for all $y\in (\frac12,\hx),\ g(y) > 0$ (since, by Lemma~\ref{lem:monotone-opinion}, the sequence $\{x_i(t)\}$ is strictly increasing), and
\item
$g(\hx) = 0$ (by definition of $\hx$).
\end{enumerate}
For the purpose of contradiction, assume that $\lim_{t\rightarrow \infty} x_i(t) = a$, where $x_0 \le a < \hx$. 
This implies, for every $\e > 0$, there exists a $t(\e)$ such that for all $t \ge t(\e),\ x_i(t+1) - x_i(t) < \e$, or equivalently, that for all $t \ge t(\e),\ g(x_i(t)) < \e$.

Let $\min_{y\in[x_0,a]} g(y) = c$. It implies for all $y\in [x_0, a],\ g(y) \ge c$. From (a), it follows that $c > 0$. Setting $\e = c$, our analysis implies the following two properties of $g$: (1) for all $t\ge 0, g(x_i(t)) \ge c$, and (2) for all $t \ge t(\e), g(x_i(t)) < c$, which contradict each other. This completes the proof by contradiction.
\end{proof}

This completes the proof of Theorem~\ref{thm:homophily-polarization}.

\section{Proofs of Section 5}\label{app:proofs-sec5}
\begin{proof}[Proof of Theorem~\ref{thm:salsa-icf-n->infty}]
\begin{lem}\label{lem:salsa}
In the limit as $n\rightarrow\infty$, SimpleSALSA  is polarizing with respect to $i$ if and only if $i$ is biased.
\end{lem}
\begin{proof}
Assume without loss of generality that $x_i > \frac12$.

Let $p_r$ be the probability that SimpleSALSA recommends a \ttr\ book. The proof consists of two steps: first we show that $p_r > \frac12$ and $p_r \le x_i$, and then we show that if $p_r > \frac12$ and $p_r \le x_i$, SimpleSALSA is polarizing with respect to $i$ if and only if $i$ is biased.
\begin{align*}
p_r &= \sum_{j\in V_2: j_2\text{ is \ttr}}\prob[i\xrightarrow{3} j]\\
&= \sum_{\substack{j_1\in N(i)\\j_1 \text{ is \ttr}}} \prob[i\xrightarrow{1} j_1] \sum_{\substack{j\in V_2\\j\text{ is \ttr}}} \prob[j_1\xrightarrow{2} j] + \sum_{\substack{j_2\in N(i)\\j_2 \text{ is \ttb}}} \prob[i\xrightarrow{1} j_2] \sum_{\substack{j\in V_2\\j\text{ is \ttr}}} \prob[j_2\xrightarrow{2} j]\\
&= \sum_{\substack{j_1\in N(i)\\j_1 \text{ is \ttr}}} \frac1{|N(i)|}\sum_{\substack{j\in V_2\\j\text{ is \ttr}}} \prob[j_1\xrightarrow{2} j] + \sum_{\substack{j_2\in N(i)\\j_2 \text{ is \ttb}}} \frac1{|N(i)|} \sum_{\substack{j\in V_2\\j\text{ is \ttr}}} \prob[j_2\xrightarrow{2} j]\\
&= \sum_{\substack{j_1\in V_2\\j_1 \text{ is \ttr}}} \frac{Z_{ij_1}}{|N(i)|}\sum_{\substack{j\in V_2\\j\text{ is \ttr}}} \prob[j_1\xrightarrow{2} j] + \sum_{\substack{j_2\in V_2\\j_2 \text{ is \ttb}}}  \frac{Z_{ij_2}}{|N(i)|} \sum_{\substack{j\in V_2\\j\text{ is \ttr}}} \prob[j_2\xrightarrow{2} j]\\
&= \sum_{\substack{j_1\in V_2\\j_1 \text{ is \ttr}}} \frac{Z_{ij_1}}{|N(i)|}\sum_{\substack{j\in V_2\\j\text{ is \ttr}}} \sum_{i'\in N(j_1)\cap N(j)} \frac{1}{|N(j_1)|}\frac{1}{|N(i')|} + \sum_{\substack{j_2\in V_2\\j_2 \text{ is \ttb}}}  \frac{Z_{ij_2}}{|N(i)|} \sum_{\substack{j\in V_2\\j\text{ is \ttr}}} \sum_{i'\in N(j_2)\cap N(j)} \frac{1}{|N(j_2)|}\frac{1}{|N(i')|} \\
&= \sum_{\substack{j_1\in V_2\\j_1 \text{ is \ttr}}} \frac{Z_{ij_1}}{|N(i)|}\sum_{\substack{j\in V_2\\j\text{ is \ttr}}} \sum_{i'\in V_1} \frac{Z_{i'j_1}Z_{i'j}}{|N(j_1)||N(i')|} + \sum_{\substack{j_2\in V_2\\j_2 \text{ is \ttb}}}  \frac{Z_{ij_2}}{|N(i)|} \sum_{\substack{j\in V_2\\j\text{ is \ttr}}} \sum_{i'\in V_1} \frac{Z_{i'j_2}Z_{i'j}}{|N(j_2)||N(i')|} \end{align*}
By Lemma~\ref{lem:ssln}, in the limit as $n\rightarrow\infty$, with probability 1,
\[
\sum_{\substack{j_1\in V_2\\j_1 \text{ is \ttr}}} \frac{Z_{ij_1}}{|N(i)|}\sum_{\substack{j\in V_2\\j\text{ is \ttr}}} \sum_{i'\in V_1} \frac{Z_{i'j_1}Z_{i'j}}{|N(j_1)||N(i')|}  \rightarrow x_i \frac{1}{k \cdot mk/2n}n\frac{mk^2(\frac14 + \var(x_1))}{n^2} = x_i\left(\frac12 + 2\var(x_1)\right)
\]
and
\[
 \sum_{\substack{j_2\in V_2\\j_2 \text{ is \ttb}}}  \frac{Z_{ij_1}}{|N(i)|} \sum_{\substack{j\in V_2\\j\text{ is \ttr}}} \sum_{i'\in V_1} \frac{Z_{i'j_2}Z_{i'j}}{|N(j_2)||N(i')|} \rightarrow (1-x_i) \frac{1}{k \cdot mk/2n}n\frac{mk^2(\frac14 - \var(x_1))}{n^2} = (1-x_i)\left(\frac12 - 2\var(x_1)\right)
\]
Therefore, in the limit as $n\rightarrow\infty$, with probability 1,
\[
p_r \rightarrow x_i\left(\frac12 + 2\var(x_1)\right) + (1-x_i)\left(\frac12 - 2\var(x_1)\right)
\]
Since $x_i > \frac12$ (by assumption), and $\var(x_1) > 0$ (by assumption), we have that
\beq\label{eq:simplesalsa-pr}
p_r > \frac12 \text{ and } p_r \le x_i
\eeq 
First, assume that $i$ is unbiased. Let $p$ be the probability that $i$ accepts the recommendation. Therefore, the probability that the recommended book was \ttr\ given that $i$ accepted the recommendation is given by
\[
\frac{p_rp}{p_rp  + (1-p_r)p} = p_r \le x_i
\]
Therefore, SimpleSALSA is not polarizing.

Now, assume that $i$ is biased. This implies $i$ accepts the recommendation of a \ttr\ book with probability $x_i$ and that of a \ttb\ book with probability $1-x_i$. Therefore, the probability that the recommended book was \ttr\ given that $i$ accepted the recommendation is given by
\[
\frac{p_rx_i}{p_rx_i + (1-x_i)(1-p_r)} > \frac{p_rx_i}{p_rx_i + p_r(1-x_i)}\text{ (since $p_r > \frac12$, from \eqref{eq:simplesalsa-pr})} = x_i
\]
Therefore, by definition, SimpleSALSA is polarizing.
\end{proof}

\begin{lem}\label{lem:icf}
In the limit as $n\rightarrow\infty$ and as $T\rightarrow\infty$, SimpleICF  is polarizing with respect to $i$ if and only if $i$ is biased.
\end{lem}
\begin{proof}
Assume without loss of generality that $x_i > \frac12$.

Let $p_r$ be the probability that SimpleICF recommends a \ttr\ book. For a node $j\in N(i)$,  let $q_{j\ttr}$ be the probability that after $T$ two-step random walks starting at $j$, the node with the largest value of $\texttt{count(j)}$, \ie, $j^*$, is \ttr, and $q_{j\ttb}$ be the corresponding probability that $j^*$ is \ttb. Then,
\begin{align*}
p_r &= \sum_{\substack{j_1\in N(i)\\j_1 \text{ is \ttr}}} \prob[i\xrightarrow{1} j_1] q_{j_1\ttr} + \sum_{\substack{j_2\in N(i)\\j_2 \text{ is \ttb}}} \prob[i\xrightarrow{1} j_2] q_{j_2\ttr}\\
&= \sum_{\substack{j_1\in N(i)\\j_1 \text{ is \ttr}}} \frac1{|N(i)|} q_{j_1\ttr} + \sum_{\substack{j_2\in N(i)\\j_2 \text{ is \ttb}}} \frac1{|N(i)|}  q_{j_2\ttr}\\
&= \sum_{\substack{j_1\in V_2\\j_1 \text{ is \ttr}}} \frac{Z_{ij_1}}{|N(i)|} q_{j_1\ttr} + \sum_{\substack{j_2\in V_2\\j_2 \text{ is \ttb}}} \frac{Z_{ij_1}}{|N(i)|}  q_{j_2\ttr}
\end{align*}
Consider $T$ two-step random walks starting at a node $j_1\in N(i)$. Observe that $q_{j_1\ttr}$ is exactly the probability that after these $T$ random walks, there exists a \ttr\ node, say $j$, such that $\texttt{count(j)} > \texttt{count(j')}$ for all \ttb\ nodes $j'$. However, as $T\rightarrow\infty$,
\[
 \prob[\text{for all \ttb\ books }j'\in V_2,\ \texttt{count(j)} > \texttt{count(j')}] = \prob[\text{for all \ttb\ books }j'\in V_2,\ \prob[j_1\xrightarrow{2}j] > \prob[j_1\xrightarrow{2} j']]
\]
since as $T\rightarrow\infty$, $\texttt{count(j)} \rightarrow T\cdot\prob[j_1\xrightarrow{2}j]$ (by the Strong Law of Large Numbers). Therefore,
\begin{align*}
q_{j_1\ttr} &= \prob[\text{for all \ttb\ books }j'\in V_2,\ \prob[j_1\xrightarrow{2}j] > \prob[j_1\xrightarrow{2} j']]
\end{align*}
Observe that for two \ttr\ books $j_1$ and $j$,
\[
\prob[j_1\xrightarrow{2}j] = \sum_{i'\in N(j_1) \cap N(j)} \frac1{|N(j_1)|}\frac1{|N(i')|} = \sum_{i'\in V_1} \frac{Z_{i'j_1}Z_{i'j}}{|N(j_1)||N(i')|}
\]
By Lemma~\ref{lem:ssln}, in the limit as $n\rightarrow \infty$, with probability 1,
\[
\prob[j_1\xrightarrow{2}j] \rightarrow \frac1{k}\frac1{mk/2n}\frac{mk^2(\frac14+\var(x_1))}{n^2} = \frac1n\left(\frac12 + 2\var(x_1)\right)
\]
Similarly, for a \ttb\ book $j'$, in the limit as $n\rightarrow \infty$, with probability 1,
\[
\prob[j_1\xrightarrow{2}j'] \rightarrow \frac1{k}\frac1{mk/2n}\frac{mk^2(\frac14-\var(x_1))}{n^2} = \frac1n\left(\frac12 - 2\var(x_1)\right)
\]
Since $\var(x_1) > 0$, in the limit as $n\rightarrow \infty,\ \prob[j_1\xrightarrow{2}j] > \prob[j_1\xrightarrow{2}j']$ with probability 1. Therefore, $q_{j_1\ttr} = 1$. By symmetry $q_{j_2\ttr} = 1- q_{j2\ttb} = 0$.
Moreover, by Lemma~\ref{lem:ssln}, in the limit as $n\rightarrow\infty,\  \sum_{\substack{j_1\in V_2\\j_1 \text{ is \ttr}}} \frac{Z_{ij_1}}{|N(i)|} = x_i$, with probability 1. 
Therefore, as $n\rightarrow\infty$,
\beq\label{eq:simpleicf-pr}
p_r = x_i
\eeq 
The rest of the analysis is identical to Lemma~\ref{lem:salsa}.
\end{proof}
This completes the proof of Theorem~\ref{thm:salsa-icf-n->infty}.
\end{proof}

\begin{proof}[Proof of Theorem~\ref{thm:ppr-n->infty}]
Assume, without loss of generality, that $x_i > \frac12$.

Let $p_r$ be the probability that SimplePPR recommends a \ttr\ book to $i$. This probability is exactly equal to the probability that after $T$ three-step random walks starting at $i$  there exists a \ttr\ node, say $j$, such that  such that $\texttt{count(j)} > \texttt{count(j')}$ for all \ttb\ nodes $j'$. However, as $T\rightarrow\infty$,
\[
 \prob[\text{for all \ttb\ books }j'\in V_2,\ \texttt{count(j)} > \texttt{count(j')}] = \prob[\text{for all \ttb\ books }j'\in V_2,\ \prob[i\xrightarrow{3}j] > \prob[i\xrightarrow{3} j']]
\]
since as $T\rightarrow\infty$, $\texttt{count(j)} \rightarrow T\cdot\prob[i\xrightarrow{3}j]$ with probability 1 (by the Strong Law of Large Numbers). Therefore,
\[
p_r = \prob[\text{for all \ttb\ books }j'\in V_2,\ \prob[i\xrightarrow{3}j] > \prob[i\xrightarrow{3} j']]
\]
For a \ttr\ book $j\in V_2$,
\begin{align*}
\prob[i\xrightarrow{3}j] &= \sum_{\substack{j_1\in N(i)\\j_1 \text{ is \ttr}}} \prob[i\xrightarrow{1} j_1]\prob[j_1\xrightarrow{2} j]  + \sum_{\substack{j_2\in N(i)\\j_2 \text{ is \ttb}}} \prob[i\xrightarrow{1} j_2]\prob[j_2\xrightarrow{2} j] \\
\prob[i\xrightarrow{3}j] &= \sum_{\substack{j_1\in N(i)\\j_1 \text{ is \ttr}}} \frac1{|N(i)|}\prob[j_1\xrightarrow{2} j]  + \sum_{\substack{j_2\in N(i)\\j_2 \text{ is \ttb}}} \frac1{|N(i)|}\prob[j_2\xrightarrow{2} j]\\
\prob[i\xrightarrow{3}j] &= \sum_{\substack{j_1\in V_2\\j_1 \text{ is \ttr}}} \frac{Z_{ij_1}}{|N(i)|}\prob[j_1\xrightarrow{2} j]  + \sum_{\substack{j_2\in V_2\\j_2 \text{ is \ttb}}} \frac{Z_{ij_2}}{|N(i)|}\prob[j_2\xrightarrow{2} j]
\end{align*}
As we showed in the proof of Lemma~\ref{lem:icf}, in the limit as $n\rightarrow \infty$,
\[
\prob[j_1\xrightarrow{2}j] \rightarrow \frac1n\left(\frac12 + 2\var(x_1)\right) \text{ and (by symmetry) }\prob[j_2\xrightarrow{2}j] \rightarrow \frac1n\left(\frac12 - 2\var(x_1)\right)
\]
 with probability 1. Moreover, by Lemma~\ref{lem:ssln}, in the limit as $n\rightarrow\infty,\  \sum_{\substack{j_1\in V_2\\j_1 \text{ is \ttr}}} \frac{Z_{ij_1}}{|N(i)|} \rightarrow x_i$, with probability 1. 
 Therefore, with probability 1,
\[
\prob[i\xrightarrow{3}j] \rightarrow \frac{x_i}{n}\left(\frac12 + 2\var(x_1)\right) + \frac{1-x_i}{n}\left(\frac12 - 2\var(x_1)\right) 
\]
Similarly, for a \ttb\ book $j'\in V_2$, in the limit as $n\rightarrow \infty$, with probability 1,
\[
\prob[i\xrightarrow{3}j'] \rightarrow \frac{x_i}{n}\left(\frac12 - 2\var(x_1)\right) + \frac{1-x_i}{n}\left(\frac12 + 2\var(x_1)\right)
\] 
Since $x_i > \frac12$ and $\var(x_1) > 0$,
\[
\prob[i\xrightarrow{3}j] > \prob[i\xrightarrow{3}j']
\]
with probability 1. In other words, $p_r = 1$. So, the probability that a book recommended by SimplePPR was \ttr\ given that it was accepted is exactly $p_r$ regardless of whether $i$ is biased or unbiased. Therefore, SimplePPR is polarizing.

\end{proof}

\end{document}